\documentclass{amsart}
\usepackage{amsmath}
\usepackage{amsfonts}

\newtheorem{theorem}{Theorem}
\theoremstyle{plain}

\newtheorem{definition}{Definition}

\newtheorem{lemma}{Lemma}

\input{tcilatex}

\begin{document}
\title[Conformable Fractional Dirac Type Integro Differential System]{%
Inverse Problems for One Dimentional Conformable Fractional Dirac Type
Integro Differential System}

\begin{abstract}
In this paper, one dimentional conformable fractional Dirac-type integro
differential system is considered. The asymptotic formulae for the
solutions, eigenvalues and nodal points are obtained. We investigate the
inverse nodal problem and give an effective procedure for solving the
inverse nodal problem with respect to given a dense subset of nodal points .
\end{abstract}

\author{Baki Keskin}
\maketitle

\section{\textbf{Introduction}}

The Dirac operator is the relativistic Schr\"{o}dinger operator in quantum
physics.The basic and comprehensive results about Dirac operators were given
in \cite{Lev}. Inverse problems for the Dirac operators have been
extensively well studied in various publications (see \cite{Alb}, \cite{Gsy1}%
, \cite{Gus}, \cite{Hor}, \cite{Ksk2}, \cite{Ksk} , \cite{wang2} and the
references therein).The subject of fractional calculus has acquired
significant popularity and major attention from several authors in various
science due mainly to its direct involvement in the problems of differential
equations in mathematics, physics, engineering and others. This topic is
initiated by (\cite{Abdel}, \cite{halil}). The fractional calculus has
gained an interesting area in mathematical research and it has attracted the
attention of several scholars. A\ variety of new results can be found (\cite%
{abu}, \cite{abu1} and references therein). In recent years, scholars have
focussed on a fractional generalization of the well known Sturm-Liouville
and Dirac problems (\cite{al}, \cite{al1}, \cite{tuna}, \cite{tuna2}, \cite%
{tuna3}, \cite{gul}, \cite{kh}, \cite{kl} and \cite{rv}).

Inverse nodal problem was started for the Sturm--Liouville operator by
McLaughlin \cite{mc1} in 1988. In 1989, Hald and McLaughlin proved that it
is sufficient to know the nodal points to determine the potential function
of the regular Sturm--Liouville problem\ with more general boundary
conditions uniquely and gave some numerical schemes for the reconstruction
of the potential from nodal points \cite{H}. Yang proposed an algorithm to
solve an inverse nodal problem for the Sturm--Liouville operator in 1997 
\cite{yang}. Such problems have been considered by several researchers in (%
\cite{Br2}, \cite{ch},, \cite{guo1}, \cite{Law}, \cite{Ozkan}, \cite{wang}, 
\cite{wang3}, \cite{wei}, \cite{Yur}, \cite{yng1}, \cite{Yang3} and \cite%
{Yang4} ) and other works.The inverse nodal problems for the Dirac operators
with various boundary conditions have been studied and shown that a dense
subset of the zeros of the first component of the eigenfunctions alone can
determine the coefficients of discussed problem \cite{Guo}, \cite{Yang2} and 
\cite{Yang5}. In \cite{mor}, the authors have developed the spectral theory
for a conformable fractional Sturm-Liouville problem and have proved
uniqueness theorem with respect to the nodal points.

Nowadays, the studies concerning the perturbation of a differential operator
by a Volterra type integral operator, namely the integro-differential
operator has acquired significant popularity and major attention from
several authors and \ take significant place in the literature.\cite{But}, 
\cite{But 2}, \cite{G}, \cite{Kur}, \cite{Kur2} and \cite{B} ).
Integro-differential operators are nonlocal, and therefore they are more
difficult for investigation, than local ones. New methods for solution of
these problems are being developed. For Sturm-Liouville type
integro-differential operators, there exist some studies about inverse
problems but there is very little study for Dirac type integro-differential
operators. The inverse nodal problem for Dirac type integro-differential
operators was first studied by \cite{kesk3}. In their study, it is shown
that the coefficients of the differential part of the operator can be
determined by using nodal points and nodal points also gives the partial
information about integral part. In \cite{kesk4}, the authors considered
Dirac type integro-differential operators with boundary conditions depend on
the spectral parameter linearly.

\section{\textbf{Conformable\ Fractional Preliminaries}}

Firstly, we want to recall some basic definitions and properties of
conformable fractional calculus which can be found in [\cite{halil}, \cite%
{Abdel}].

\begin{definition}
Let $f:[0,\infty )\rightarrow 
\mathbb{R}
$ be a given function. Then the conformable fractional derivative of $f$ of
order $\alpha $ is defined by:

$D_{x}^{\alpha }f(x)=\underset{\in \rightarrow 0}{\lim }\dfrac{f(x+\epsilon
x^{1-\alpha })-f(t)}{\epsilon }\alpha ,$ $D_{x}^{\alpha }f(0)=\underset{%
t\rightarrow 0^{+}}{\lim }D^{\alpha }f(x),$\newline
for all $x>0,$ $\alpha \in (0,1].$ If this limit exist and finite at $x_{0},$
we say $f$ is $\alpha -$differentiable at $x_{0}.$ Note that if $f$ is
differentiable, then $D^{\alpha }f(x)=x^{1-\alpha }f^{\prime }(x).$
\end{definition}

\begin{definition}
The conformable fractional Integral starting from $0$ of order $\alpha $ is
defined by

$I_{\alpha }f(x)=\int_{0}^{x}f(t)d_{\alpha }t=\int_{0}^{x}t^{\alpha
-1}f(t)dt,$\newline
for all $x>0.$
\end{definition}

\begin{lemma}
Let $f:[a,\infty )\rightarrow 
\mathbb{R}
$ be any continuous function. Then, for all $x>a$, we have $D_{x}^{\alpha
}I_{\alpha }f(x)=f(x).$
\end{lemma}

\begin{lemma}
Let $f:(a,b)\rightarrow 
\mathbb{R}
$ be any differentiable function. Then, for all $x>a$, we have $I_{\alpha
}D_{x}^{\alpha }f(x)=f(x)-f(a).$
\end{lemma}

\begin{theorem}
\textbf{(}$\alpha -$\textbf{integration by parts): }Let $f,g:[a,b]%
\rightarrow 
\mathbb{R}
$ be two conformable fractional differentiable functions. Then,

$\int_{a}^{b}f(x)D_{x}^{\alpha }g(x)d_{\alpha
}x=f(b)g(b)-f(a)g(a)-\int_{a}^{b}g(x)D_{x}^{\alpha }f(x)d_{\alpha }x$
\end{theorem}

\begin{theorem}
($\alpha -$\textbf{Leibnitz Rule}) Let $t^{\alpha -1}f(x,t)$ and $t^{\alpha
-1}f_{x}(x,t)$ be continuous in $x$ on some regions of the $(x,t)-$plane,
including $a(x)\leq t\leq b(x),$ $x_{0}\leq x\leq x_{1}.$ If $a(x)$ and $%
b(x) $ are both $\alpha -$differentiable for $x_{0}\leq x\leq x_{1},$ then

$D_{x}^{\alpha }\left[ \int_{a(x)}^{b(x)}f(x,t)d_{\alpha }t\right]
=f(x,b(x))b^{\alpha -1}(x)D_{x}^{\alpha }b(x)-f(x,a(x))a^{\alpha
-1}(x)D_{x}^{\alpha }a(x)+\int_{a(x)}^{b(x)}D_{x}^{\alpha }f(x,t)d_{\alpha
}t.$
\end{theorem}

\begin{definition}
The space $C_{\alpha }^{n}[a,b]$ consists of all functions defined on the
interval $[a,b]$ which are continuously $\alpha -$differentiable up to order 
$n.$
\end{definition}

\section{\textbf{Conformable Fractional Dirac Systems}}

In this work, we consider the following one-dimentional conformable
fractional Dirac type integro-differential system%
\begin{equation}
BY+\Omega (x)Y+\int\limits_{0}^{x}M(x,t)Yd_{\alpha }t=\lambda Y,\text{ \ }%
x\in (0,\pi ),
\end{equation}%
with the boundary conditions%
\begin{eqnarray}
\sin \theta y_{1}(0)+\cos \theta y_{2}(0) &=&0\medskip  \\
\sin \beta y_{1}(\pi )+\cos \beta y_{2}(\pi ) &=&0
\end{eqnarray}%
where, $\theta $ is a real constant, $\lambda $ is the spectral parameter, 
\newline
$B=\left( 
\begin{array}{cc}
0 & D_{x}^{\alpha } \\ 
-D_{x}^{\alpha } & 0%
\end{array}%
\right) ,$ \ $\Omega (x)=\left( 
\begin{array}{cc}
p(x) & 0 \\ 
0 & r(x)%
\end{array}%
\right) ,$ $M(x,t)=\left( 
\begin{array}{cc}
M_{11}(x,t) & M_{12}(x,t) \\ 
M_{21}(x,t) & M_{22}(x,t)%
\end{array}%
\right) ,$ \ $Y=\left( 
\begin{array}{c}
y_{1} \\ 
y_{2}%
\end{array}%
\right) ,$ $p(x)$, $r(x)$, $M\left( x,t\right) $ are real-valued conformable
fractional differentiable functions and $x^{\alpha -1}p(x)$ and $x^{\alpha
-1}r(x)$ are continuous on $(0,\pi )$.

Let $\varphi (x,\lambda )=\left( \varphi _{1}(x,\lambda ),\varphi
_{2}(x,\lambda )\right) ^{T}$ be the solution of (1) satisfying the initial
condition $\varphi (0,\lambda )=(\cos \theta ,-\sin \theta )^{T}$. It is
clear that $\varphi (x,\lambda )$ is an entire function of $\lambda $
satisfies the following conformable fractional Volterra integral equations:%
\begin{equation}
\begin{array}{l}
\varphi _{1}(x,\lambda )=\cos \theta \cos \left( \lambda \dfrac{x^{\alpha }}{%
\alpha }\right) +\sin \theta \sin \left( \lambda \dfrac{x^{\alpha }}{\alpha }%
\right) \medskip \\ 
+\int_{0}^{x}\sin \left( \lambda \dfrac{x^{\alpha }-t^{\alpha }}{\alpha }%
\right) p(t)\varphi _{1}(t,\lambda )d_{\alpha }t+\int_{0}^{x}\cos \left(
\lambda \dfrac{x^{\alpha }-t^{\alpha }}{\alpha }\right) r(t)\varphi
_{2}(t,\lambda )d_{\alpha }t\medskip \\ 
+\int_{0}^{x}\int_{0}^{t}\sin \left( \lambda \dfrac{x^{\alpha }-t^{\alpha }}{%
\alpha }\right) \left\{ M_{11}(t,\xi )\varphi _{1}(\lambda ,\xi
)+M_{12}(t,\xi )\varphi _{2}(\lambda ,\xi )\right\} d_{\alpha }\xi d_{\alpha
}t\medskip \\ 
+\int_{0}^{x}\int_{0}^{t}\cos \left( \lambda \dfrac{x^{\alpha }-t^{\alpha }}{%
\alpha }\right) \left\{ M_{21}(t,\xi )\varphi _{1}(\lambda ,\xi
)+M_{22}(t,\xi )\varphi _{2}(\lambda ,\xi )\right\} d_{\alpha }\xi d_{\alpha
}t%
\end{array}%
\end{equation}%
\begin{equation}
\begin{array}{l}
\varphi _{2}(x,\lambda )=\cos \theta \sin \left( \lambda \dfrac{x^{\alpha }}{%
\alpha }\right) -\sin \theta \cos \left( \lambda \dfrac{x^{\alpha }}{\alpha }%
\right) \medskip \\ 
-\int_{0}^{x}\cos \left( \lambda \dfrac{x^{\alpha }-t^{\alpha }}{\alpha }%
\right) p(t)\varphi _{1}(t,\lambda )d_{\alpha }t+\int_{0}^{x}\sin \left(
\lambda \dfrac{x^{\alpha }-t^{\alpha }}{\alpha }\right) r(t)\varphi
_{2}(t,\lambda )d_{\alpha }t\medskip \\ 
-\int_{0}^{x}\int_{0}^{t}\cos \left( \lambda \dfrac{x^{\alpha }-t^{\alpha }}{%
\alpha }\right) \left\{ M_{11}(t,\xi )\varphi _{1}(\lambda ,\xi
)+M_{12}(t,\xi )\varphi _{2}(\lambda ,\xi )\right\} d_{\alpha }\xi d_{\alpha
}t\medskip \\ 
+\int_{0}^{x}\int_{0}^{t}\sin \left( \lambda \dfrac{x^{\alpha }-t^{\alpha }}{%
\alpha }\right) \left\{ M_{21}(t,\xi )\varphi _{1}(\lambda ,\xi
)+M_{22}(t,\xi )\varphi _{2}(\lambda ,\xi )\right\} d_{\alpha }\xi d_{\alpha
}t%
\end{array}%
\end{equation}

\begin{lemma}
Let $\ f(x)$ be a function in $C_{\alpha }^{1}[0,\pi ],$ then

$\underset{\left\vert \lambda \right\vert \rightarrow \infty }{\lim }\exp
(-\left\vert \func{Im}\lambda \dfrac{\pi ^{\alpha }}{\alpha }\right\vert
)\int_{0}^{\pi }f(x)\cos \lambda \dfrac{x^{\alpha }}{\alpha }d_{\alpha }x=0$

and

$\underset{\left\vert \lambda \right\vert \rightarrow \infty }{\lim }\exp
(-\left\vert \func{Im}\lambda \dfrac{\pi ^{\alpha }}{\alpha }\right\vert
)\int_{0}^{\pi }f(x)\sin \lambda \dfrac{x^{\alpha }}{\alpha }d_{\alpha }x=0$%
\bigskip

\begin{theorem}
For $\left\vert \lambda \right\vert \rightarrow \infty ,$ the following
asymptotic formulae are valid:%
\begin{eqnarray}
&&\left. \varphi _{1}(x,\lambda )=\cos \left[ \lambda \dfrac{x^{\alpha }}{%
\alpha }-\mu (x)-\theta \right] +\frac{1}{2\lambda }\upsilon (x)\cos \left[
\lambda \dfrac{x^{\alpha }}{\alpha }-\mu (x)-\theta \right] \right. \medskip
\notag \\
&&-\frac{1}{2\lambda }\upsilon (0)\cos \left[ \lambda \dfrac{x^{\alpha }}{%
\alpha }-\mu (x)+\theta \right] +\frac{1}{2\lambda }\sin \left[ \lambda 
\dfrac{x^{\alpha }}{\alpha }-\mu (x)-\theta \right] \int_{0}^{x}\upsilon
^{2}(t)d_{\alpha }t\medskip \\
&&-\frac{1}{2\lambda }K\left( x\right) \cos \left[ \lambda \dfrac{x^{\alpha }%
}{\alpha }-\mu (x)-\theta \right] -\frac{1}{2\lambda }L\left( x\right) \sin %
\left[ \lambda \dfrac{x^{\alpha }}{\alpha }-\mu (x)-\theta \right] \medskip 
\notag \\
&&+o\left( \frac{1}{\lambda }\exp (\left\vert \tau \right\vert \dfrac{%
x^{\alpha }}{\alpha })\right) ,  \notag
\end{eqnarray}%
\begin{eqnarray}
&&\left. \varphi _{2}(x,\lambda )=\sin \left[ \lambda \dfrac{x^{\alpha }}{%
\alpha }-\mu (x)-\theta \right] -\frac{1}{2\lambda }\upsilon (x)\sin \left[
\lambda \dfrac{x^{\alpha }}{\alpha }-\mu (x)-\theta \right] \right. \medskip
\notag \\
&&-\frac{1}{2\lambda }\upsilon (0)\sin \left[ \lambda \dfrac{x^{\alpha }}{%
\alpha }-\mu (x)+\theta \right] -\frac{1}{2\lambda }\cos \left[ \lambda 
\dfrac{x^{\alpha }}{\alpha }-\mu (x)-\theta \right] \int_{0}^{x}\upsilon
^{2}(t)d_{\alpha }t\medskip \\
&&-\frac{1}{2\lambda }K\left( x\right) \sin \left[ \lambda \dfrac{x^{\alpha }%
}{\alpha }-\mu (x)-\theta \right] +\frac{1}{2\lambda }L\left( x\right) \cos %
\left[ \lambda \dfrac{x^{\alpha }}{\alpha }-\mu (x)-\theta \right] \medskip 
\notag \\
&&+o\left( \frac{1}{\lambda }\exp (\left\vert \tau \right\vert \dfrac{%
x^{\alpha }}{\alpha })\right) ,  \notag
\end{eqnarray}%
uniformly in $x\in \lbrack 0,\pi ],$ where, $\mu (x)=\dfrac{1}{2}%
\int_{0}^{x}(p(t)+r(t))d_{\alpha }t,$ $\upsilon (x)=\dfrac{1}{2}\left(
p(x)-r(x)\right) ,$ $K(x)=\int_{0}^{x}(M_{11}(t,t)-M_{22}(t,t))d_{\alpha }t,$
$L(x)=\int_{0}^{x}(M_{12}(t,t)-M_{21}(t,t))d_{\alpha }t$ \newline
and $\tau =\func{Im}\lambda .$
\end{theorem}
\end{lemma}

\begin{proof}
we denote 
\begin{eqnarray*}
&&\left. \varphi _{1,0}(x,\lambda )=\cos \left( \lambda \dfrac{x^{\alpha }}{%
\alpha }-\theta \right) ,\medskip \right. \\
&&\left. \varphi _{1,n+1}(x,\lambda )=\int_{0}^{x}\sin \left( \lambda \dfrac{%
x^{\alpha }-t^{\alpha }}{\alpha }\right) p(t)\varphi _{1,n}(t,\lambda
)d_{\alpha }t\right. \\
&&+\int_{0}^{x}\cos \left( \lambda \dfrac{x^{\alpha }-t^{\alpha }}{\alpha }%
\right) r(t)\varphi _{2,n}(t,\lambda )d_{\alpha }t\medskip \\
&&+\int_{0}^{x}\int_{0}^{t}\sin \left( \lambda \dfrac{x^{\alpha }-t^{\alpha }%
}{\alpha }\right) \left\{ M_{11}(t,\xi )\varphi _{1,n}(\lambda ,\xi
)+M_{12}(t,\xi )\varphi _{2,n}(\lambda ,\xi )\right\} d_{\alpha }\xi
d_{\alpha }t\medskip \\
&&+\int_{0}^{x}\int_{0}^{t}\cos \left( \lambda \dfrac{x^{\alpha }-t^{\alpha }%
}{\alpha }\right) \left\{ M_{21}(t,\xi )\varphi _{1,n}(\lambda ,\xi
)+M_{22}(t,\xi )\varphi _{2,n}(\lambda ,\xi )\right\} d_{\alpha }\xi
d_{\alpha }t,\medskip
\end{eqnarray*}%
\begin{eqnarray*}
&&\left. \varphi _{2,0}(x,\lambda )=\sin \left( \lambda \dfrac{x^{\alpha }}{%
\alpha }-\theta \right) ,\medskip \right. \\
&&\left. \varphi _{2,n+1}(x,\lambda )=-\int_{0}^{x}\cos \left( \lambda 
\dfrac{x^{\alpha }-t^{\alpha }}{\alpha }\right) p(t)\varphi _{1,n}(t,\lambda
)d_{\alpha }t\right. \medskip \\
&&+\int_{0}^{x}\sin \left( \lambda \dfrac{x^{\alpha }-t^{\alpha }}{\alpha }%
\right) r(t)\varphi _{2,n}(t,\lambda )d_{\alpha }t\medskip \\
&&-\int_{0}^{x}\int_{0}^{t}\cos \left( \lambda \dfrac{x^{\alpha }-t^{\alpha }%
}{\alpha }\right) \left\{ M_{11}(t,\xi )\varphi _{1,n}(\lambda ,\xi
)+M_{12}(t,\xi )\varphi _{2,n}(\lambda ,\xi )\right\} d_{\alpha }\xi
d_{\alpha }t\medskip \\
&&+\int_{0}^{x}\int_{0}^{t}\sin \left( \lambda \dfrac{x^{\alpha }-t^{\alpha }%
}{\alpha }\right) \left\{ M_{21}(t,\xi )\varphi _{1,n}(\lambda ,\xi
)+M_{22}(t,\xi )\varphi _{2,n}(\lambda ,\xi )\right\} d_{\alpha }\xi
d_{\alpha }t,
\end{eqnarray*}%
applying successive approximations method to the equations (4) and (5) and
using Lemma3, we get the estimates (6) and (7)$,$
\end{proof}

\bigskip

The characteristic function $\Delta (\lambda )$ of the problem (1)-(3) is
defined by the relation 
\begin{equation}
\Delta (\lambda )=\varphi _{1}(\pi ,\lambda )\sin \beta +\varphi _{2}(\pi
,\lambda )\cos \beta ,
\end{equation}%
It is obvious that $\Delta (\lambda )$ is an entire function and its zeros,
namely $\left\{ \lambda _{n}\right\} _{n\in 
\mathbb{Z}
}$ ,\ coincide with the eigenvalues of the problem (1)-(3). Using the
asymptotic formulae (6) and (7), one can easily obtain%
\begin{eqnarray}
&&\left. \Delta (\lambda )=\sin \left[ \lambda \dfrac{x^{\alpha }}{\alpha }%
-\mu (x)-\theta +\beta \right] -\frac{1}{2\lambda }\upsilon (x)\sin \left[
\lambda \dfrac{x^{\alpha }}{\alpha }-\mu (x)-\theta -\beta \right] \right.
\medskip  \notag \\
&&-\frac{1}{2\lambda }\upsilon (0)\sin \left[ \lambda \dfrac{x^{\alpha }}{%
\alpha }-\mu (x)+\theta +\beta \right] \medskip  \notag \\
&&-\frac{1}{2\lambda }\cos \left[ \lambda \dfrac{x^{\alpha }}{\alpha }-\mu
(x)-\theta +\beta \right] \int_{0}^{x}\upsilon ^{2}(t)d_{\alpha }t\medskip \\
&&-\frac{1}{2\lambda }K\left( x\right) \sin \left[ \lambda \dfrac{x^{\alpha }%
}{\alpha }-\mu (x)-\theta +\beta \right] \medskip  \notag \\
&&+\frac{1}{2\lambda }L\left( x\right) \cos \left[ \lambda \dfrac{x^{\alpha }%
}{\alpha }-\mu (x)-\theta +\beta \right] +o\left( \frac{1}{\lambda }\exp
(\left\vert \tau \right\vert \dfrac{x^{\alpha }}{\alpha })\right) ,  \notag
\end{eqnarray}%
for sufficiently large $\left\vert \lambda \right\vert .$ Since the
eigenvalues of the problem (1)-(3) are the roots of $\Delta (\lambda _{n})=0$%
, we can write the following equation for them:%
\begin{equation*}
\begin{array}{l}
\left( 1-\dfrac{1}{2\lambda _{n}}\upsilon (\pi )\cos 2\beta -\dfrac{1}{%
2\lambda _{n}}\upsilon (0)\cos 2\theta -\dfrac{1}{2\lambda _{n}}K(\pi
)\right) \tan (\lambda _{n}\dfrac{\pi ^{\alpha }}{\alpha }-\mu (\pi )-\theta
+\beta )=\medskip \\ 
-\dfrac{1}{2\lambda _{n}}\upsilon (\pi )\sin 2\beta -\dfrac{1}{2\lambda _{n}}%
\upsilon (0)\sin 2\theta +\dfrac{1}{2\lambda _{n}}\int_{0}^{\pi }\upsilon
^{2}(t)d_{\alpha }t-\dfrac{1}{2\lambda _{n}}L(\pi )+o\left( \dfrac{1}{%
\lambda _{n}}\right) \medskip%
\end{array}%
\end{equation*}%
which implies that%
\begin{equation*}
\begin{array}{l}
\tan (\lambda _{n}\dfrac{\pi ^{\alpha }}{\alpha }-\mu (\pi )-\theta +\beta
)=\medskip \\ 
\left( 1-\dfrac{1}{2\lambda _{n}}\upsilon (\pi )\cos 2\beta -\dfrac{1}{%
2\lambda _{n}}\upsilon (0)\cos 2\theta -\dfrac{1}{2\lambda _{n}}K(\pi
)\right) ^{-1}\times \medskip \\ 
\times \left( -\dfrac{1}{2\lambda _{n}}\upsilon (\pi )\sin 2\beta -\dfrac{1}{%
2\lambda _{n}}\upsilon (0)\sin 2\theta +\dfrac{1}{2\lambda _{n}}%
\int_{0}^{\pi }\upsilon ^{2}(t)d_{\alpha }t-\dfrac{1}{2\lambda _{n}}L(\pi
)+o\left( \dfrac{1}{\lambda _{n}}\right) \right)%
\end{array}%
\end{equation*}%
for sufficiently large $n.$

We obtain from the last equation,%
\begin{eqnarray}
&&\left. \lambda _{n}=\dfrac{\alpha }{\pi ^{\alpha -1}}n+\alpha \dfrac{%
\theta +\mu (\pi )-\beta }{\pi ^{\alpha }}\medskip \right.  \notag \\
&&+\dfrac{\alpha }{2n\pi ^{\alpha }}\left( \upsilon (\pi )\sin 2\beta
-\upsilon (0)\sin 2\theta +\int_{0}^{\pi }\upsilon ^{2}(t)d_{\alpha }t-L(\pi
)\right) \medskip \\
&&+o\left( \dfrac{1}{n}\right)  \notag
\end{eqnarray}%
for $\left\vert n\right\vert \rightarrow \infty .$

\section{\textbf{Main Results}}

In this section, we obtain the asymptotic formula for the nodal points of
considered problem and prove an inverse nodal problem for the
one-dimensional conformable fractional Dirac-type integro differential
system.

\begin{lemma}
For sufficiently large $n$, the first component $\varphi _{1}(x,\lambda
_{n}) $ of the eigenfunction $\varphi (x,\lambda _{n})$ has exactly $n$
nodes $\left\{ x_{n}^{j}:j=\overline{0,n-1}\right\} $ in the interval $%
\left( 0,\pi \right) $:\newline
$0<x_{n}^{0}<x_{n}^{1}<...<x_{n}^{n-1}<\pi $. The numbers $\left\{
x_{n}^{j}\right\} $ satisfy the following asymptotic formula:%
\begin{equation}
\begin{array}{l}
\left( x_{n}^{j}\right) ^{\alpha }=\dfrac{\left( j+1/2\right) \pi ^{\alpha }%
}{n}+\dfrac{\mu (x_{n}^{j})+\theta }{n\pi ^{1-\alpha }}\medskip \\ 
-\dfrac{\left( j+1/2\right) \pi ^{\alpha }}{n}\left( \dfrac{\theta +\mu (\pi
)-\beta }{n\pi }\right) -\dfrac{\theta +\mu (\pi )-\beta }{\pi ^{2-\alpha
}n^{2}}\left( \mu (x_{n}^{j})+\theta \right) \medskip \\ 
+\dfrac{\alpha }{2n^{2}}\left( \upsilon (0)\sin 2\theta
+\int_{0}^{x_{n}^{j}}\upsilon ^{2}(t)d_{\alpha }t-L\left( x_{n}^{j}\right)
\right) +o\left( \frac{1}{n^{2}}\right) .%
\end{array}%
\end{equation}
\end{lemma}

\begin{proof}
From (6), the following asymptotic formula can be written for sufficiently
large $n$%
\begin{eqnarray*}
&&\left. \varphi _{1}(x,\lambda _{n})=\cos \left[ \lambda _{n}\dfrac{%
x^{\alpha }}{\alpha }-\mu (x)-\theta \right] +\frac{1}{2\lambda _{n}}%
\upsilon (x)\cos \left[ \lambda _{n}\dfrac{x^{\alpha }}{\alpha }-\mu
(x)-\theta \right] \right. \medskip \\
&&-\frac{1}{2\lambda _{n}}\upsilon (0)\cos \left[ \lambda _{n}\dfrac{%
x^{\alpha }}{\alpha }-\mu (x)+\theta \right] +\frac{1}{2\lambda _{n}}\sin %
\left[ \lambda _{n}\dfrac{x^{\alpha }}{\alpha }-\mu (x)-\theta \right]
\int_{0}^{x}\upsilon ^{2}(t)d_{\alpha }t\medskip \\
&&-\frac{1}{2\lambda _{n}}K\left( x\right) \cos \left[ \lambda _{n}\dfrac{%
x^{\alpha }}{\alpha }-\mu (x)-\theta \right] -\frac{1}{2\lambda _{n}}L\left(
x\right) \sin \left[ \lambda _{n}\dfrac{x^{\alpha }}{\alpha }-\mu (x)-\theta %
\right] \medskip \\
&&+o\left( \frac{1}{\lambda _{n}}\exp (\left\vert \tau \right\vert \dfrac{%
x^{\alpha }}{\alpha })\right) ,
\end{eqnarray*}%
from $\varphi _{1}(\left( x_{n}^{j}\right) ^{\alpha },\lambda _{n})=0,$ we
get%
\begin{eqnarray*}
&&\left. \cos \left[ \lambda _{n}\dfrac{\left( x_{n}^{j}\right) ^{\alpha }}{%
\alpha }-\mu (x_{n}^{j})-\theta \right] =-\frac{1}{2\lambda _{n}}\upsilon
(x_{n}^{j})\cos \left[ \lambda _{n}\dfrac{\left( x_{n}^{j}\right) ^{\alpha }%
}{\alpha }-\mu (x_{n}^{j})-\theta \right] \right. \medskip \\
&&+\frac{1}{2\lambda _{n}}\upsilon (0)\cos \left[ \lambda _{n}\dfrac{\left(
x_{n}^{j}\right) ^{\alpha }}{\alpha }-\mu (x_{n}^{j})-\theta \right] \cos
2\theta \medskip \\
&&-\frac{1}{2\lambda _{n}}\upsilon (0)\sin \left[ \lambda _{n}\dfrac{\left(
x_{n}^{j}\right) ^{\alpha }}{\alpha }-\mu (x_{n}^{j})-\theta \right] \sin
2\theta \medskip \\
&&-\frac{1}{2\lambda _{n}}\sin \left[ \lambda _{n}\dfrac{\left(
x_{n}^{j}\right) ^{\alpha }}{\alpha }-\mu (x_{n}^{j})-\theta \right]
\int_{0}^{x_{n}^{j}}\upsilon ^{2}(t)d_{\alpha }t\medskip \\
&&+\frac{1}{2\lambda _{n}}K\left( x_{n}^{j}\right) \cos \left[ \lambda _{n}%
\dfrac{\left( x_{n}^{j}\right) ^{\alpha }}{\alpha }-\mu (x_{n}^{j})-\theta %
\right] \medskip \\
&&+\frac{1}{2\lambda _{n}}L\left( x_{n}^{j}\right) \sin \left[ \lambda _{n}%
\dfrac{\left( x_{n}^{j}\right) ^{\alpha }}{\alpha }-\mu (x_{n}^{j})-\theta %
\right] +o\left( \frac{1}{\lambda _{n}}\right) ,
\end{eqnarray*}%
$\medskip $%
\begin{eqnarray*}
&&\left. \tan \left[ \lambda _{n}\dfrac{\left( x_{n}^{j}\right) ^{\alpha }}{%
\alpha }-\mu (x_{n}^{j})-\theta -\frac{\pi }{2}\right] \left( 1+\frac{1}{%
2\lambda _{n}}\upsilon (x_{n}^{j})-\frac{1}{2\lambda _{n}}\upsilon (0)\cos
2\theta -\frac{1}{2\lambda _{n}}K\left( x_{n}^{j}\right) \right) =\right.
\medskip \\
&&\frac{1}{2\lambda _{n}}\upsilon (0)\sin 2\theta +\frac{1}{2\lambda _{n}}%
\int_{0}^{x_{n}^{j}}\upsilon ^{2}(t)d_{\alpha }t-\frac{1}{2\lambda _{n}}%
L\left( x_{n}^{j}\right) +o\left( \frac{1}{\lambda _{n}}\right) ,
\end{eqnarray*}%
\newline
Taking into account Taylor's expansion formula for the arctangent, we get$%
\medskip $ 
\begin{equation*}
\lambda _{n}\dfrac{\left( x_{n}^{j}\right) ^{\alpha }}{\alpha }-\mu
(x_{n}^{j})-\theta -\frac{\pi }{2}=j\pi +\dfrac{1}{2\lambda _{n}}\left(
\upsilon (0)\sin 2\theta +\int_{0}^{x_{n}^{j}}\upsilon ^{2}(t)d_{\alpha
}t-L\left( x_{n}^{j}\right) \right) +o\left( \dfrac{1}{\lambda _{n}}\right)
.\medskip
\end{equation*}%
It follows from the last equality

\begin{equation*}
\dfrac{\left( x_{n}^{j}\right) ^{\alpha }}{\alpha }=\dfrac{\left( j+\frac{1}{%
2}\right) \pi +\mu (x_{n}^{j})+\theta }{\lambda _{n}}+\dfrac{1}{2\lambda
_{n}^{2}}\left( \upsilon (0)\sin 2\theta +\int_{0}^{x_{n}^{j}}\upsilon
^{2}(t)d_{\alpha }t-L\left( x_{n}^{j}\right) \right) +o\left( \dfrac{1}{%
\lambda _{n}^{2}}\right) .
\end{equation*}%
The relation (11) is proven by using the asymptotic formula%
\begin{equation*}
\lambda _{n}^{-1}=\dfrac{\pi ^{\alpha -1}}{n\alpha }\left\{ 1-\dfrac{\mu
(\pi )+\theta -\beta }{n\pi }-\dfrac{\left( \upsilon (\pi )\sin 2\beta
-\upsilon (0)\sin 2\theta +\int_{0}^{\pi }\upsilon ^{2}(t)d_{\alpha }t-L(\pi
)\right) }{2n^{2}\pi }+o\left( \dfrac{1}{n^{2}}\right) \right\}
\end{equation*}
\end{proof}

Let $X$ be the set of nodal points and $\mu (\pi )=0.$ For each fixed $x\in
\left( 0,\pi \right) $ and $\alpha \in (0,1]$ $\ $we can choose a sequence $%
\left( x_{n}^{j}\right) \subset X$ so that $x_{n}^{j}$ converges to $x.$
Then the following limits are exist and finite:%
\begin{equation*}
\underset{\left\vert n\right\vert \rightarrow \infty }{\lim }n\left( \left(
x_{n}^{j(n)}\right) ^{\alpha }-\dfrac{\left( j+1/2\right) \pi ^{\alpha }}{n}%
\right) =f(x),
\end{equation*}%
where%
\begin{equation}
f(x)=\dfrac{\mu (x)+\theta }{\pi ^{1-\alpha }}-\dfrac{x^{\alpha }}{\pi }%
\left( \theta +\mu (\pi )-\beta \right)
\end{equation}%
and%
\begin{equation*}
\underset{\left\vert n\right\vert \rightarrow \infty }{\lim }2n^{2}\left(
\left( x_{n}^{j(n)}\right) ^{\alpha }-\dfrac{\left( j+1/2\right) \pi
^{\alpha }}{n}-\dfrac{\mu (x_{n}^{j})+\theta }{n\pi ^{1-\alpha }}+\dfrac{%
\left( j+1/2\right) \pi ^{\alpha }}{n}\left( \dfrac{\theta +\mu (\pi )-\beta 
}{n\pi }\right) \right) =g(x),
\end{equation*}%
where%
\begin{equation}
g(x)=\alpha \left( \upsilon (0)\sin 2\theta +\int_{0}^{x}\upsilon
^{2}(t)d_{\alpha }t-L\left( x\right) \right)
\end{equation}%
Therefore, proof of the following theorem is clear. Let $\mu (\pi )=0$

\begin{theorem}
The given dense subset of nodal points $X$ uniquely determines the
coefficients $\theta $ and $\beta $ of the boundary conditions and if $L(x)$
is known, $X$ also uniquely determines the potential $\Omega (x)$ a.e. on $%
\left( 0,\pi \right) $ . Moreover, $\Omega (x),$ $L(x),$ $\theta $ and $%
\beta $ can be reconstructed by the following formulae:

\textbf{Step-1:} For each fixed $x\in (0,\pi )$ and $\alpha \in (0,\pi ],$
choose a sequence $\left( x_{n}^{j(n)}\right) \subset X$ such that $\underset%
{\left\vert n\right\vert \rightarrow \infty }{\lim }x_{n}^{j(n)}=x;$

\textbf{Step-2: }Find the function $f(x)$ from (12) and calculate 
\begin{eqnarray}
\theta &=&f(0)  \notag \\
\beta &=&f(\pi )\theta ^{1-\alpha }  \notag \\
D_{x}^{\alpha }\mu (x) &=&D_{x}^{\alpha }f(x)+\dfrac{\alpha }{\pi }\left[
f(0)-f(\pi )f(0)^{1-\alpha }\right]
\end{eqnarray}%
\textbf{Step-3: }Find the function $g(x)$ from (13) and calculate 
\begin{equation}
\medskip \upsilon (x)=\dfrac{1}{\sqrt{\alpha }}\sqrt{D_{x}^{\alpha
}(g(x)+\alpha L(x))}
\end{equation}%
\textbf{Step-4: } If $L(x)$ is known then from (14) and (15) calculate%
\begin{eqnarray*}
p(x) &=&\upsilon (x)+D_{x}^{\alpha }\mu (x) \\
r(x) &=&D_{x}^{\alpha }\mu (x)-\upsilon (x)
\end{eqnarray*}%
If $p(x)$ and $r(x)$ are known then from (13) calculate%
\begin{equation*}
L\left( x\right) =\upsilon (0)\sin 2\theta +\int_{0}^{x}\upsilon
^{2}(t)d_{\alpha }t-\frac{g(x)}{\alpha }
\end{equation*}
\end{theorem}

\end{document}